\documentclass[12pt]{article}

\usepackage[T1]{fontenc}
\usepackage{lmodern}
\usepackage[margin=1in]{geometry}
\usepackage{amsmath,amssymb,amsthm}
\usepackage[authoryear,round]{natbib}
\usepackage{setspace}
\usepackage[colorlinks,citecolor=blue,linkcolor=blue,urlcolor=blue]{hyperref}
\usepackage{etoolbox}

\theoremstyle{plain}

\newtheorem{theorem}{Theorem}
\newtheorem{proposition}{Proposition}
\newtheorem{lemma}{Lemma}
\newtheorem{assumption}{Assumption}

\apptocmd{\appendix}{}{}{}

\title{A Reverse War of Attrition\thanks{We gratefully acknowledge valuable comments on the early concept from Olivier Blanchard, Piotr Dworczak, Paul Klemperer, David K. Levine, Maurice Obstfeld, Sergio Rebelo, Ariel Rubinstein, Tomas Sjöström, and Ivan Werning. An early exposition of the idea was provided by \citet{stanczak2025sticky} prior to the present co-authorship. All remaining errors are ours.}}
\author{Kazimierz Stańczak\thanks{Independent Researcher, Warsaw, Poland. Email: \texttt{stanczakkazimierz@gmail.com}}
  \and Paweł Struski\thanks{University of Warsaw; Group for Research in Applied Economics. Email: \texttt{p.struski@uw.edu.pl}}}
\date{September 23, 2026}

\begin{document}

\maketitle

\begin{abstract}
We formulate a complete-information two-player strategic timing game that we refer to as the reverse war of attrition (RWoA). Waiting yields a strictly positive current rent, and the player that outwaits its rival may earn an even higher flow. However, remaining in the waiting state when the rival has moved triggers a discrete late-mover penalty. The distinctive primitive is a player-owned productive buffer which enables the rent-generating activity while waiting, but is degraded by this activity. A rival’s move can raise the waiting rent and accelerate depletion. We characterize a unique symmetric equilibrium. We apply RWoA to Bertrand competition in a homogeneous-good market, where vintage inventory serves as the buffer clock for timing repricing. We also discuss competitive maintenance. As an extension, we replicate the bilateral game across a continuum of disjoint heterogeneous local markets and, via the exact law of large numbers, derive deterministic aggregate state shares.


\end{abstract}

\thispagestyle{empty}
\newpage
\setcounter{page}{1}

\section{Introduction}\label{sec:1}

When is it rational to wait before making a move that is known eventually to be necessary? Retailers delay raising prices as long as they can still sell inherited inventory whose replacement cost has increased. Firms postpone maintenance or replacement, which would temporarily shut down or reduce production, while the machines or technologies in use continue to generate profits. Incumbents continue harvesting legacy products before introducing innovations that cannibalize them. In each case, the decision is not whether to move, but when to move.

The standard timing-game benchmarks point in several directions. In a war of attrition (WoA), waiting is costly, but the player who outwaits the rival obtains the prize. In preemption, moving first creates a leader advantage. RWoA has neither the first-mover advantage of preemption nor the unqualified last-mover advantage of classical WoA. In our setup, moving first gives up a positive waiting rent. The player who moves second may actually earn a higher current flow after the rival moves, but in our two-player game that player is also the last mover and incurs a discrete buffer-dependent late-mover penalty. RWoA thus combines a local second-mover advantage in current flow with a
last-mover disadvantage.

The central primitive in RWoA is a finite, player-owned productive buffer that acts as a self-depleting endogenous clock. This buffer may be inventory, remaining machine life, an unused AI compute allocation, or available contracted capacity. It permits the rent-generating activity to continue in the waiting state, but that very activity
progressively exhausts or degrades its stock. Thus the same asset both enables a profitable delay and restricts its duration. Buffers, like inventories or the remaining useful life, are standard in economics; the distinctive use here is to make the productive buffer the strategic clock that links profitable waiting to the discipline on waiting too
long.

The RWoA mechanism seems especially transparent in Bertrand-style price competition. Consider two competing retailers, selling an identical good, each holding vintage inventory at the moment when the replacement cost of inventory has increased. While both retain the legacy price, each earns a positive margin. If one reprices first, the remaining
low-price seller captures diverted demand and earns more profit. Yet that diverted demand accelerates inventory depletion. If the remaining stock is too small, the firm may be unable to fulfill redirected orders, risking customer dissatisfaction from ``empty shelves'' at a low advertised price. Inventory is thus simultaneously the source of the
waiting rent and the clock governing its length. Directionally, the same logic governs the competitive timing of maintenance, where equilibrium strategies critically depend on the remaining useful life. Other candidate applications include nonrenewable reserved AI compute or test capacity and exhaustible reserves supporting a legacy technology.

In what follows, we first study the canonical RWoA mechanism as a two-player complete-information commitment timing game. Each player privately draws at date zero a stopping date and commits to that date for the adjustment episode. Until then it earns a positive flow payoff. If the rival stops first, the remaining player---the second and, by
construction, also the last mover---may earn a higher flow, but the rival's action also triggers a buffer-dependent late-mover penalty. Once a player stops, the transition to the adjusted state is irreversible within the episode, and its continuation value in the timing game is normalized to zero, without loss of generality in the sense made precise in Section \ref{sec:3.1}.

Within this ex-ante committed-date class, no pure-strategy equilibrium exists. The unique symmetric equilibrium we construct is atomless and mixed. A player scheduled to move strictly first can profitably delay by an instant while remaining first, whereas a common positive date is vulnerable to an epsilon-earlier deviation under the random-priority tie convention. The equilibrium has connected support. Its density solves a linear Volterra equation, and a shooting condition selects the unique support start that generates one unit of probability by the terminal waiting boundary.

The commitment assumption is substantive. There are many examples where firms fix dates in advance and changing them is costly. Retailers pre-schedule promotion/repricing windows, often tied to supplier/customer commitments; manufacturers book maintenance shutdowns, crews, and parts; firms schedule IT migrations before legacy systems cease to be supported by suppliers. So a firm may rationally stick to its date even if the rival moves first.

The bilateral model also admits a useful aggregation step under a continuum of disjoint binary markets. We allow these local binary markets to differ in inventories, demand, margins, depletion rates, or other primitives. Each market first solves its own bilateral RWoA and has a type-specific equilibrium timing CDF. Under the exact law of large numbers, individual stopping times remain random, but cross-sectional state shares are deterministic integrals of those type-specific probabilities.

We do not claim any novelty for profitable delay, second-mover advantage, endogenous timing, mixed stopping, or a Volterra or ODE characterization as such. Early and important analyses of second-mover advantages and endogenous timing include \citet{reinganum1985two} and \citet{hamilton1990endogenous}; more recent related work includes \citet{amir2006second}, Smirnov and Wait \citeyearpar{smirnov2015innovation,smirnov2021preemption}, and \citet{karp2025timing}. Nor do we claim that a local second-mover advantage can generate attrition-type behavior: \citet{steg2015symmetric} establishes this in a general stochastic timing setting. Nor do we claim novelty for the bare fact that a timing game may be neither a pure preemption game nor a pure war of attrition: \citet{agastya2006continuing} and \citet{park2008caller} already study spanning classes in which both forces coexist.

Our proposed, more narrowly defined novelty has two distinct layers. First, RWoA combines familiar timing forces into a particular payoff geometry: a positive contemporaneous rent from remaining in the pre-adjustment state, a rival move that can raise that rent, and a last-mover disadvantage that ultimately disciplines delay. Second, the economic mechanism that ties those forces together is a finite, player-owned productive buffer. The buffer itself---inventory, remaining useful life, capacity, or another productive stock---is of course not a new economic object. What is new, and what we emphasize as the distinctive primitive of RWoA, is its use as a self-depleting strategic clock: the stock is required for the profitable waiting activity, that activity runs the stock down, and the rival's move can both raise the return to waiting and make the clock run faster or increase the exposure it creates. To our knowledge,
prior timing-game work has not used a productive buffer in the dual role studied here. In this precise sense, RWoA reverses the primitive waiting-flow sign of classical attrition and locally inverts the preemption force emphasized by \citet{fudenberg1985preemption}. The ingredients are standard separately; our claim is that their combination, together with the productive buffer clock that connects them, defines a distinct timing-game geometry. The RWoA label is not
our contribution: \citet{seel2016reverse} had already used ``The Reverse War of Attrition'' for a completely different game. We retain the label because it transparently captures the waiting-flow sign reversal of WoA that we model, while making no priority claim over the name.

The paper proceeds as follows. Section \ref{sec:2} locates RWoA relative to wars of attrition, preemption, beneficial delay, and second-mover models, as well as theories of sales, inventory, and capacity-limited price competition. Section \ref{sec:3} defines the two-player game. Section \ref{sec:4} establishes the deviation logic and support properties. Section \ref{sec:5} constructs the unique symmetric equilibrium in the committed-date class and shows that it is atomless and absolutely continuous. Section \ref{sec:6} presents the pricing and maintenance applications. Section \ref{sec:7} studies heterogeneous continuum replication. Section \ref{sec:8} concludes.

\section{Related Work}\label{sec:2}

\subsection{War of Attrition and Preemption}\label{sec:2.1}

In the textbook WoA developed by \citet{smith1973logic}, \citet{smith1974theory}, \citet{hendricks1988war},
\citet{bulow1999generalized}, and others, each player wants the other to concede. Waiting imposes a negative flow cost, but the player who outwaits the rival obtains the prize. Equilibrium strategic mixing balances the marginal cost of another instant of delay against the probability that the rival concedes. At the level of the primitive
waiting flow, RWoA reverses the sign of the classical WoA waiting payoff.

In their foundational contribution, \citet{fudenberg1985preemption} provide the canonical model of preemption: a leader advantage creates a threat of preemption and strategic pressure toward earlier adoption. RWoA locally inverts that force. Moving first gives up a positive waiting rent, while the rival's move can increase the current payoff from remaining in the waiting state. Strategic pressure therefore points toward later action until the buffer-dependent late-mover penalty supplies the countervailing force.

RWoA's distinguishing payoff geometry is local. Preemption has a first-mover advantage. Classical WoA rewards the player who outwaits the rival, paying the flow cost while waiting. RWoA instead gives the second mover a local advantage in current flow, while imposing a discrete disadvantage on that same player as the last mover. We model the coexistence of a second-mover flow advantage and a last-mover disadvantage.

That coexistence is not itself new. \citet{agastya2006continuing} study a continuing war of attrition, and \citet{park2008caller} introduce a spanning class of timing games with rank-order payoffs---a lump-sum reward that depends only on a player's ordinal stopping rank---subsuming both the war of attrition and the preemption game under complete information and unobserved actions. RWoA falls outside that class because its payoffs are not rank-measurable. \citet{anderson2017rushes} study a continuum-player timing game with payoff growth over time, in which the payoff depends on the stopping date and the stopping quantile. There, a continuation game is summarized by elapsed time and the fraction that has already stopped, so it does not matter when predecessors moved. By contrast, in RWoA, both the continuation rent and the late-mover exposure are indexed by the rival's date.

\subsection{Beneficial Delay, Second-Mover Advantage, and Priority Risk}\label{sec:2.2}

Beneficial delay, follower advantages, and endogenous timing are well studied. \citet{reinganum1985two} develops a two-stage R\&D model with endogenous second-mover advantages. \citet{hamilton1990endogenous} provide a foundational endogenous-timing framework in which players choose whether to act early or wait. \citet{dutta1995better}
show that later adoption can be advantageous when product quality improves. \citet{amir2006second} study second-mover advantage and endogenous price leadership in a Bertrand duopoly. Smirnov and Wait \citeyearpar{smirnov2015innovation,smirnov2021preemption} analyze
general timing games in which second-mover advantages coexist with rent equalization or even preemption. \citet{karp2025timing} develop a general continuous-time two-player framework for the timing and order of moves, illustrated with price-competition applications. These contributions concern follower advantages, endogenous order, or general timing; RWoA does not claim novelty for any of these objects.

RWoA is most directly related to \citet{steg2015symmetric}, who shows that in a general stochastic timing game a local second-mover advantage may generate a WoA with mixed stopping rates. Therefore, our construct ``second-mover advantage implies attrition'' is not new. RWoA instead imposes a particular economic structure behind that timing force: a positive operating rent exists already while both players wait; the rival's move can raise that rent; and a self-depleting
productive buffer both supports the rent-generating activity and creates the late-mover exposure that prevents
unbounded delay.

A related strand obtains a switch between preemption and attrition regimes from uncertainty. \citet{hoppe2000second} shows that technological uncertainty can turn a preemption game into a waiting game as the probability that adoption is profitable falls, so that second-mover advantages arise from informational spillovers. By contrast, in RWoA, agents act under complete information and the balance between attrition and preemption is governed by the productive buffer.

Very recent work by \citet{fudenberg2026racing} obtains war-of-attrition behavior in a dynamic stopping game under imperfect monitoring; their mechanism is distinct from the productive-buffer clock studied here. 

Among mixed-delay models, RWoA is especially adjacent to \citet{bobtcheff2017researcher}. In their model,
researchers wait to mature private ideas but risk losing priority if the rival discloses first. Their symmetric
two-player equilibrium is mixed and characterized by a differential equation selected through a global boundary
condition. RWoA differs in the source and direction of the payoff force. Waiting yields an explicit contemporaneous
operating rent; the rival's move can raise that rent; and a finite player-owned buffer both supports waiting and
determines the late-mover penalty.

Our equilibrium concept is deliberate rather than a substitute for an unproved refinement. The canonical RWoA is an ex-ante committed-date timing game. Precommitment is an established benchmark in the foundational timing literature: \citet{reinganum1981diffusion} studies technology adoption under precommitment; \citet{fudenberg1985preemption} explicitly contrast their richer timing formalism with Reinganum's precommitment equilibria; \citet{hendricks1988war} characterize Nash equilibrium outcomes in the continuous-time complete-information WoA; and \citet{hamilton1990endogenous} distinguish an observable-delay game from an action-commitment game. By contrast, \citet{laraki2005continuous} and \citet{riedel2017subgame} develop explicitly history-contingent subgame-perfect formulations. SPE is therefore the equilibrium concept for a different game form in which plans can be revised after observed history, not a refinement that the present fixed-date game is intended to satisfy.

\subsection{Productive Buffer Clocks, Irreversibility, and Endogenous Timing}\label{sec:2.3}

A calendar deadline is an exogenous clock: it passes independently of the player's actions. RWoA instead focuses on an endogenous economic clock---a productive stock used by the very activity that makes waiting valuable. Natural examples include vintage inventory, remaining useful life, and exhaustible productive capacity. The economic activity of
waiting therefore runs its own clock.

The novelty claim about the buffer is intentionally narrow. Inventories, machine life, capacity constraints, and exhaustible stocks are standard in economics and operations research. We are not aware, however, of prior timing-game work that makes such a player-owned productive stock the endogenous strategic clock in this dual sense: it is necessary for profitable waiting, it is depleted or degraded by that waiting, and rival action can simultaneously increase the waiting rent and accelerate depletion or the associated late-mover exposure. RWoA therefore does two things at once: it combines familiar strategic timing forces, and it adds a concrete productive-buffer clock that
links them economically.

A recent strand also puts a state variable into the war of attrition, but an exogenous one. In \citet{georgiadis2022absence} the players' flow payoffs while fighting vary with stochastic market conditions; in \citet{gieczewski2025evolving} a public state moves the two players' fighting costs in opposite directions. In both, the state is moved by nature, not by either player's own operation. The RWoA state is instead a player-owned productive stock that finances the waiting rent and is drawn down by the very activity that earns it, faster once the rival has moved.

The buffer must also be distinguished from the fighting fund of the budget-constrained war of attrition. In \citet{dekel2007jump}, \citet{foster2018wars}, and \citet{huangfu2023resource}, a finite resource pays the cost of staying in the contest, and exhausting it forces concession. The RWoA buffer finances a rent rather than a cost.

The one-way transition from waiting to adjustment is also familiar from real-options reasoning. \citet{dixit1994investment} emphasize that when an action is irreversible or costly to reverse, waiting preserves an option that is exercised by committing; \citet{grenadier2002option} embeds option exercise in strategic competition. RWoA uses this broad
commitment logic but does not require exogenous uncertainty. Irreversibility is therefore not itself our novelty claim. The distinctive primitive is the productive buffer clock: a state variable depleted by the same activity that makes waiting profitable and that, after rival action, can run down faster or generate greater late-mover exposure.

Except for its strategic timing game core, our model is naturally related to single-agent operations research. \citet{vanderduynschouten1995maintenance}, for example, condition preventive maintenance on both machine age and intermediate-buffer content in a production system. RWoA instead asks how a rival's action changes the return to continued operation and the depletion or exposure of a player-owned buffer. Strategic interaction may thereby replace
a deterministic intervention rule with mixed timing.

\subsection{Sales, Limited Capacity, and Promotional Pricing}\label{sec:2.4}

The pricing application is related to several adjacent lines of literature. \citet{varian1980model} provides a canonical
equilibrium theory of sales in which strategic pricing generates randomized prices. \citet{gelman1983judo} show how credible capacity limitation---and coupons representing claims on scarce discounted output---can alter the competitive response to a low price. \citet{guimaraes2011sales} show that sales can be strategic substitutes. Inventory competition itself is also well established: \citet{lippman1997competitive} study the competitive newsvendor problem,
while \citet{mahajan2001inventory} analyze inventory competition with dynamic consumer substitution after stockouts. Those models make inventory strategic, but inventory is chosen as a decision variable; in our pricing application, inherited inventory is instead the state variable that supports profitable waiting and clocks the timing of repricing.

RWoA is adjacent but distinct from these mechanisms. Our object is not a static price distribution, entry accommodation, or the frequency of temporary sales. A finite productive buffer supports continued operation at the legacy low price; when the rival reprices, residual demand and the return to remaining low can rise, while the same buffer is depleted faster and determines the exposure to the late-mover disadvantage. The strategic variable is therefore the timing of exit from the low-price state. The same buffer both enables profitable waiting and disciplines its duration: inventory is not merely a constraint on the game but the critical endogenous clock of the game.

\section{The Binary Reverse War of Attrition}\label{sec:3}

\subsection{Players, Timing, and the Buffer}\label{sec:3.1}

There are two symmetric players, $i \in \{1,2\}$. Time is continuous for $0 < t \le T$, where $0 < T < \infty$ is the terminal waiting boundary. Excluding the mathematical origin rules out instantaneous action at the instant the game
opens. Each player begins in a waiting state $W$ and chooses a stopping date at which to move to an adjusted state $A$. A pure strategy is a date $\tau_i$ satisfying $0 < \tau_i \le T$. At date zero, players use mixed strategies over
stopping dates and privately draw their dates. Primitives and distributions are common knowledge; the draws are not observed until actions occur. There is no exogenous uncertainty in payoff primitives or state dynamics; randomization
comes from mixed strategies and, off path, the tie-breaking convention specified below.

Our Theorem \ref{thm1} treats $T$ as the maximal admissible commitment horizon. In a buffer application, $T$ can itself be induced by robust technological feasibility---for example, as the supremum of committed dates that remain feasible for every possible earlier rival move. The equilibrium construction needs only a finite $T$ and therefore takes this horizon in reduced form.

The irreversibility assumption concerns the state transition, not the timing commitment. Before the selected stopping date the player remains in the waiting state; once it moves to the adjusted state, it cannot switch back within the
same adjustment episode. Separately, the stopping date selected at date zero is committed for the episode. We use this benchmark to isolate the productive-buffer mechanism, not because real firms can never revise a plan. Repricing
calendars, promotion windows, IT or advertising updates, supplier and customer commitments, booked maintenance slots, crews, and shutdown plans all make ex post replanning potentially costly. This interpretation has a simple robustness
property: if changing a previously selected date after any history requires a fixed revision cost $\kappa$, compactness of $[0,T]$ and bounded flow payoffs and penalties imply a finite uniform upper bound on the gross gain from any revision. Hence, for sufficiently large $\kappa$, no ex post revision is profitable and the committed-date equilibrium survives in the enlarged game with a revision option. At zero revision cost the continuation problem is
genuinely different: after a rival move has triggered a sunk late-mover exposure, positive post-rival flow can create an incentive to reoptimize, so the equilibrium below need not be subgame perfect in that enlarged, costless-revision
game. Exact SPE is not guaranteed in general continuous-time timing games under weak assumptions---\citet{laraki2005continuous} guarantee Markov subgame-perfect $\varepsilon$-equilibria for every $\varepsilon>0$ and show that stronger exact-existence conclusions require additional structure. We abstract from that richer feedback game in order to keep the main RWoA mechanism transparent.

The one-way transition has a direct economic interpretation in the applications. After a permanent increase in marginal or replacement cost, repricing from a legacy price to the post-shock price completes the adjustment to the
new cost environment; absent a new cost shock, a later return to the legacy price is a new pricing episode rather than an undoing of the same adjustment. Likewise, once a productive asset, like a machine, is withdrawn and maintenance
begins, shutdown, setup, inspection, disassembly, and associated downtime or sunk costs make an immediate switchback economically distinct; completion of maintenance and return to service begins a new operating cycle. This is the same
broad commitment logic emphasized in real-options models of irreversible action \citep{dixit1994investment}.

Once a player moves, its continuation value in the timing payoff is normalized to zero. More precisely, for every fixed rival strategy we may subtract any payoff component that is independent of the player's own stopping date; under
a mixed rival this subtracts its expectation. Such a relative-payoff normalization leaves best responses and Nash equilibria unchanged. Section \ref{sec:6.1}  uses this normalization explicitly.

Subscripts $0$ and $1$ below denote whether the rival has not yet moved or has moved; they are not player indices. The waiting state is supported by a finite buffer. Before the rival moves, the buffer follows
\begin{equation}\label{eq:buffer-pre}
\dot{b}_0(t) = -q_0\bigl(b_0(t),t\bigr), \qquad b_0(0) = b^{\mathrm{init}} > 0,
\end{equation}
where $q_0(b,t) > 0$ is the depletion rate and $b^{\mathrm{init}}$ is the initial buffer level. If the rival moves at date $s$, the remaining player's buffer path may change to
\begin{equation}\label{eq:buffer-post}
\frac{\partial b_1(t;s)}{\partial t} = -q_1\bigl(b_1(t;s),t\bigr), \qquad b_1(s;s) = J\bigl(b_0(s),s\bigr),
\end{equation}
where $J$ is a buffer-reset map satisfying $J(b,s) \le b$, and $q_1(b,t) > 0$ is the post-rival depletion rate. In applications with accelerated depletion, $q_1(b,t) \ge q_0(b,t)$ on the relevant state space. Equations (\ref{eq:buffer-pre})--(\ref{eq:buffer-post}) motivate the reduced-form objects $g_0$, $g_1$, $\Lambda$, and $T$; the equilibrium results below take those objects as primitives.

For formal well-posedness of the buffer construction, assume that $q_0$ and $q_1$ admit extensions to $\mathbb{R} \times [0,T]$ that are continuous in $(b,t)$ and globally Lipschitz-continuous in $b$, with a Lipschitz
constant uniform in $t$. Assume also that $J$ is continuous and maps the economically relevant buffer domain into itself, and restrict $T$ so that the resulting paths remain in that domain. Standard Cauchy-problem theory then yields
a unique pre-rival path $b_0$ and, for every rival stopping date $s$, a unique post-rival path $b_1(\cdot\,;s)$, with continuous dependence on the initial state; continuity in $s$ follows from continuity of the reset map and of the ODE
flow in its initial time and state \citep[Theorem 2.2.2, p.~90]{kolokoltsov2019differential}.

\subsection{Waiting Flows and the Late-Mover Penalty}\label{sec:3.2}

While both players wait, each receives the positive flow payoff
\[
g_0(t) = G_0\bigl(b_0(t),t\bigr) > 0.
\]
If the rival moved at date $s < t$ and player $i$ continues waiting, player $i$ receives
\[
g_1(t;s) = G_1\bigl(b_1(t;s),t\bigr) > 0.
\]
The benchmark rival-enhanced-waiting condition is
\begin{equation}\label{eq:rival-enhanced}
g_1(t;s) \ \ge\ g_0(t) \qquad \text{for } 0 < s < t \le T.
\end{equation}
Thus the rival's move weakly raises the current payoff from continued waiting. This is the model's local second-mover advantage in current flow. It is not a claim that the second mover's total pairwise payoff always exceeds that of the first mover, because the second mover also bears the late-mover penalty defined below. The inequality is imposed throughout the canonical equilibrium analysis in Sections \ref{sec:4} and \ref{sec:5}; applications outside that class may relax it.

If the rival moves first at date $s$, the remaining player---who is now the second and last mover---incurs the late-mover penalty
\[
\Lambda(s) \equiv \mathcal{L}\bigl(b_0(s),s\bigr) > 0.
\]
Here $\mathcal{L}$ is a penalty map evaluated at the actual rival-action date and the prevailing pre-rival buffer. The late-mover penalty may be literal or the date-$s$ present value of an irreversible exposure: unfilled commitments,
emergency procurement, downtime, durable customer loss incurred by offering ``empty shelves'', failure, or obsolescence. Post-rival flows included in $g_1$ are excluded from $\Lambda$, so there is no double counting. A positive late-mover penalty need not imply that the remaining player's total payoff is always below that of the first mover.

The trigger convention is substantive. Once the rival moves first at date $s$, the exposure $\Lambda(s)$ is locked in at that date: moving an instant later does not retroactively erase it. Economically, $\Lambda(s)$ can be the date-$s$
present value of commitments, emergency actions, or durable customer or operating losses set in motion by the rival's move. This discrete order-of-moves exposure is what makes being second and last costly. If instead the late-mover loss vanished continuously as the player's own stopping date approached $s$ from above, the no-pure and no-atom arguments below would require a different model.

\subsection{Pairwise Payoffs}\label{sec:3.3}

Let $r \ge 0$ be the discount rate. Let player $i$ choose $t$ and player $j$ choose $s$. Ignoring the tie case for the moment, player $i$'s payoff is
\[
\Pi_i(t,s) = \int_0^{t} e^{-ru} g_0(u)\, du, \qquad t < s,
\]
and
\[
\Pi_i(t,s) = \int_0^{s} e^{-ru} g_0(u)\, du + \int_s^{t} e^{-ru} g_1(u;s)\, du - e^{-rs}\Lambda(s), \qquad s < t.
\]
At an exact tie, random priority assigns each player the first-mover outcome with probability one-half and the late-mover outcome with probability one-half. Because the post-rival interval then has zero length,
\begin{equation}\label{eq:tie-payoff}
\Pi_i(t,t) = \int_0^{t} e^{-ru} g_0(u)\, du - \tfrac{1}{2} e^{-rt}\Lambda(t).
\end{equation}
Random priority is a tie-breaking convention, not uncertainty about the economic primitives. It matters only at exact ties: the no-pure and no-atom arguments require a strictly positive expected late-mover exposure at a tie. The
equilibrium constructed below is atomless, so tie payoffs are off path and do not enter the Volterra equation or the equilibrium density. If exact simultaneity instead eliminated the late-mover exposure with probability one,
Propositions \ref{prop1} and \ref{prop2} would not follow as stated.

\subsection{Expected Payoff Against a Mixed Rival}\label{sec:3.4}

Let the rival use an atomless CDF $F$ over dates satisfying $0 < t \le T$. Atomlessness makes endpoint conventions immaterial. The expected payoff from choosing $t$ is
\begin{equation}\label{eq:expected-payoff}
U(t;F) = \int_0^{t} e^{-ru}\left[\bigl(1 - F(u)\bigr) g_0(u) + \int_0^{u} g_1(u;s)\, dF(s)\right] du
- \int_0^{t} e^{-rs} \Lambda(s)\, dF(s).
\end{equation}
If $F$ has density $f$, differentiation at a continuity point gives
\begin{equation}\label{eq:payoff-diff}
\frac{\partial U(t;F)}{\partial t} = e^{-rt}\left[\bigl(1 - F(t)\bigr) g_0(t) + \int_0^{t} g_1(t;s) f(s)\, ds
- \Lambda(t) f(t)\right].
\end{equation}
Discounting cancels from the local indifference condition because all marginal terms are evaluated at the same date. It still enters $\Lambda$ whenever the late-mover penalty capitalizes future consequences. The deviation arguments
below use continuity and boundedness of the primitives on their compact domains; Assumption \ref{assump1} states the regularity conditions used for equilibrium construction.

\section{Deterministic Timing, Atoms, and Support}\label{sec:4}

\begin{proposition}[No positive-time pure-strategy equilibrium]\label{prop1}
Under the random-priority tie convention (\ref{eq:tie-payoff}), suppose $g_0$ and $\Lambda$ are continuous and strictly positive on $[0,T]$. No pure-strategy Nash equilibrium exists in the committed-date game. Moreover, no common positive
date can be implemented by a nonbinding agreement.
\end{proposition}

\begin{proof}
If $t_i < t_j$, player $i$ can delay to $t_i + \varepsilon < t_j$. It remains the first mover, incurs no late-mover penalty, and earns an additional strictly positive flow. The case $t_j < t_i$ is symmetric. If $t_i = t_j = t > 0$,
each player bears the discrete expected late-mover penalty $e^{-rt}\Lambda(t)/2$. Moving to $t - \varepsilon$ eliminates that exposure while sacrificing only $O(\varepsilon)$ flow. The deviation is profitable for sufficiently small $\varepsilon$. The same deviation defeats a nonbinding recommendation to stop jointly at $t$.
\end{proof}

\begin{proposition}[No positive-time atoms]\label{prop2}
Suppose $g_0$ and $\Lambda$ are continuous and strictly positive on $[0,T]$, and $g_1$ extends continuously to the closed triangle $0 \le s \le t \le T$. Under the random-priority tie convention \textup{(\ref{eq:tie-payoff})}, a symmetric equilibrium
cannot contain an atom at any $t>0$.
\end{proposition}

\begin{proof}
Suppose the rival places mass $p > 0$ at $t$. At a symmetric atom, choosing $t$ must be optimal, but on the atom event it entails the expected late-mover penalty $p e^{-rt}\Lambda(t)/2$. Choose a sequence $\varepsilon_n \downarrow 0$
such that $t - \varepsilon_n > 0$ and the rival has no atom at $t - \varepsilon_n$; this is possible because a probability distribution has at most countably many atoms. A deviation to $t - \varepsilon_n$ eliminates the fixed
tie-event loss. By continuity and boundedness of the primitives, and because the probability assigned to $(t - \varepsilon_n, t)$ tends to zero, every other payoff change tends to zero. Hence the deviation is profitable for all sufficiently large $n$.
\end{proof}

\begin{lemma}[Connected support]\label{lem1}
Suppose $g_0$ and $\Lambda$ are continuous and strictly positive on $[0,T]$, and $g_1$ extends continuously to the closed triangle with $g_1(t;s)>0$. In any symmetric atomless mixed-strategy equilibrium, the closure of the support is
an interval $[t_0,T]$ for some $t_0$ satisfying $0 \le t_0 < T$.
\end{lemma}

\begin{proof}
Suppose an open gap $(\alpha,\beta)$ separates two support components. On the gap, the rival distribution assigns no mass. Differentiating (\ref{eq:expected-payoff}) there leaves only the expected current-flow term, which is strictly positive because $g_0$
and $g_1$ are positive. Payoff therefore rises across the gap, so both boundary points cannot yield the common equilibrium payoff. The support is connected. Let $t^U$ be its upper endpoint. If $t^U < T$, choosing a date just above $t^U$ leaves the late-mover penalty unchanged---because the rival stops by $t^U$ almost surely---but adds strictly positive post-rival flow. Hence $t^U = T$. Assumption \ref{assump2} below will imply $t_0 > 0$.
\end{proof}

\section{The Symmetric Mixed-Strategy Equilibrium}\label{sec:5}

\subsection{The Volterra Indifference Equation}\label{sec:5.1}

By Proposition \ref{prop2}, any symmetric equilibrium is atomless; by Lemma \ref{lem1}, the closure of its support is $[t_0,T]$. Let $F$ be such a candidate equilibrium, with $F(t) = 0$ for $t < t_0$, $F(t_0) = 0$, and $F(T) = 1$. Expected payoff $U(\cdot\,;F)$ is continuous, and the support indifference condition therefore makes it constant on the entire support. In (\ref{eq:expected-payoff}), the expected-flow component is absolutely continuous in the chosen date, whereas the late-mover-penalty component has Stieltjes measure $e^{-rt}\Lambda(t)\, dF(t)$. Hence this weighted Stieltjes measure must be absolutely continuous with respect to $dt$. Because $e^{-rt}\Lambda(t)$ is continuous and bounded away from zero on $[0,T]$, $F$ is absolutely continuous. Writing $f = dF/dt$, for almost every date in the interior of the support, indifference requires
\begin{equation}\label{eq:indifference}
\Lambda(t) f(t) = \bigl(1 - F(t)\bigr) g_0(t) + \int_{t_0}^{t} g_1(t;s) f(s)\, ds.
\end{equation}
Using $F(t) = \int_{t_0}^{t} f(s)\, ds$, this becomes the linear Volterra equation
\begin{equation}\label{eq:volterra}
f(t) = a(t) + \int_{t_0}^{t} K(t,s) f(s)\, ds,
\end{equation}
where
\[
a(t) = \frac{g_0(t)}{\Lambda(t)}, \qquad K(t,s) = \frac{g_1(t;s) - g_0(t)}{\Lambda(t)}.
\]
Under the benchmark condition, $K$ is nonnegative. The right-hand side of (\ref{eq:volterra}) is continuous, so the density admits a continuous version. The density needed to discipline waiting equals the baseline flow-to-penalty ratio plus the
additional density induced by the higher post-rival waiting flow.

\subsection{Assumptions and Shooting}\label{sec:5.2}

\begin{assumption}[Regular binary RWoA]\label{assump1}
(i) $g_0$ and $\Lambda$ are continuous on $[0,T]$, with $g_0(t) > 0$ and $\Lambda(t) > 0$. (ii) $g_1$ extends
continuously to the closed triangle $0 \le s \le t \le T$ and satisfies $g_1(t;s) \ge g_0(t)$. We study independent
mixed strategies over committed dates satisfying $0 < t \le T$ and characterize symmetric equilibria.
\end{assumption}

For shooting purposes, extend the candidate lower boundary to $\tau = 0$. For each $\tau \in [0,T]$, let $f_\tau$ be the unique continuous solution of
\[
f_\tau(t) = a(t) + \int_{\tau}^{t} K(t,s) f_\tau(s)\, ds, \qquad t \in [\tau,T].
\]
Existence and uniqueness for every $\tau$ follow from standard linear Volterra theory; Appendix \ref{appA1} gives a self-contained successive-approximation proof (see also \citealp[Chapter 1]{brunner2017volterra}, and
\citealp[Theorem 2.1.1, pp.~86--87]{kolokoltsov2019differential}). Define
\[
F_\tau(t) = \int_{\tau}^{t} f_\tau(s)\, ds, \qquad M(\tau) = F_\tau(T).
\]

\begin{assumption}[Sufficient probability mass]\label{assump2}
$M(0) > 1$. Since $K \ge 0$ implies $f_0 \ge a$, a transparent sufficient condition is $\int_0^{T} g_0(t)/\Lambda(t)\, dt > 1$.
\end{assumption}

The condition says that if mixing began arbitrarily close to zero, the indifference equation would generate more than one unit of probability by $T$. The equilibrium support must therefore begin strictly later.

\begin{lemma}[Shooting map]\label{lem2}
Under Assumption \ref{assump1}, $M(\tau)$ is continuous and strictly decreasing in $\tau$, with $M(T) = 0$.
\end{lemma}

\begin{proof}
Continuity follows from the explicit Lipschitz estimate in Appendix \ref{appA2}. If $\tau_1 < \tau_2$, the solution beginning at $\tau_1$ accumulates strictly positive density on $[\tau_1,\tau_2]$. On $[\tau_2,T]$, the difference between the two solutions satisfies a Volterra equation with nonnegative forcing and nonnegative kernel, so it is nonnegative. Hence $M(\tau_1) > M(\tau_2)$. Appendix \ref{appA2} gives the details.
\end{proof}

\begin{theorem}[Unique symmetric equilibrium in committed stopping dates]\label{thm1}
Under Assumptions \ref{assump1} and \ref{assump2}, there is a unique $t_0$ satisfying $0 < t_0 < T$ and $M(t_0) = 1$. The CDF
$F^{\star}(t) = 0$ for $t < t_0$ and $F^{\star}(t) = \int_{t_0}^{t} f_{t_0}(s)\, ds$ for $t \in [t_0,T]$ is a symmetric Nash equilibrium of the admissible ex-ante committed-date game. It is atomless, absolutely continuous, and has connected support. It is the unique symmetric equilibrium in this strategy class; no claim is made here about asymmetric equilibria.
\end{theorem}

\begin{proof}
Lemma \ref{lem2}, Assumption \ref{assump2}, and the intermediate value theorem give a unique $t_0$ with $M(t_0) = 1$. Positivity of $a$ and nonnegativity of $K$ imply $f_{t_0} > 0$, so $F^{\star}$ is continuous and strictly increasing from zero to one. On the support, equation (\ref{eq:indifference}), equivalently (\ref{eq:volterra}), makes the payoff derivative in (\ref{eq:payoff-diff}) equal to zero almost everywhere; because expected payoff is absolutely continuous, it is therefore constant on the support. Before $t_0$, the rival never moves and the derivative is $e^{-rt} g_0(t) > 0$, so every earlier date is strictly worse; dates after $T$ are infeasible. For uniqueness, consider any other symmetric equilibrium. Proposition \ref{prop2} makes it atomless; Lemma \ref{lem1} gives connected support with some lower endpoint; Section \ref{sec:5.1} makes it absolutely continuous with a density satisfying the corresponding Volterra equation. A lower endpoint of zero would require $M(0) = 1$, contradicting Assumption \ref{assump2}. Hence its lower endpoint is positive, and strict monotonicity of the shooting map selects the same support start and density.
\end{proof}

Note that the assumptions of Theorem \ref{thm1} are conditions on the reduced-form objects $g_0$, $g_1$, $\Lambda$, and $T$. Theorem \ref{thm1} therefore applies whether or not those objects are generated by a productive buffer; the buffer's role is to provide the economic mechanism that generates them, as illustrated in Section \ref{sec:6}.

\subsection{Markov-Separable Benchmark}\label{sec:5.3}

Suppose the post-rival flow depends on the current date but not on the exact rival stopping date: $g_1(t;s) = g_1(t)$. Then
\begin{equation}\label{eq:markov-sep-ode}
\dot{F}(t) = \frac{g_0(t) + \bigl[g_1(t) - g_0(t)\bigr] F(t)}{\Lambda(t)}, \qquad F(t_0) = 0, \quad F(T) = 1.
\end{equation}
The local indifference condition has the direct interpretation
\[
\text{marginal expected waiting flow} \;=\; \text{late-mover penalty at } t \;\times\; \text{rival density at } t.
\]
At a fixed date and CDF value, a larger late-mover penalty lowers the density needed to discipline waiting. The induced change in the support start is a global object determined by the terminal mass condition; comparative statics
must therefore account for the boundary condition, not only the ODE's right-hand side.

\section{Applications of the RWoA Productive Buffer Clock}\label{sec:6}

\subsection{Bertrand Pricing with Vintage Inventory}\label{sec:6.1}

Consider two retailers selling a homogeneous product with perfectly inelastic total demand $D > 0$. The legacy price is $p_L$ and the adjusted price is $p_H > p_L$. Let $m_L > 0$ denote per-unit contribution on vintage inventory sold
at $p_L$, and let $m_H \ge 0$ denote per-unit contribution once both firms have adjusted and split demand at $p_H$. If one firm has adjusted while its rival remains at $p_L$, homogeneous-good Bertrand demand goes to the low-price seller,
so the adjusted firm sells none of this product until the rival also adjusts. To map this product market exactly into the normalized timing payoff of Section \ref{sec:3}, fix the rival stopping date $s$ and subtract the payoff the player would obtain by adjusting immediately at date zero against that same $s$. This benchmark is independent of the player's own chosen stopping date and therefore does not affect best responses; under a mixed rival, its expectation is likewise an irrelevant constant.

While both firms remain at the legacy price, each serves $D/2$. If the rival reprices first, the remaining low-price seller serves $D$ until it also reprices; after both have repriced, each again serves $D/2$. Under the normalization just described, the pre-rival and post-rival waiting flows are
\[
g_0 = \frac{m_L D}{2}, \qquad g_1 = m_L D - \frac{m_H D}{2} = (1+\delta) g_0, \qquad
\delta \equiv 1 - \frac{m_H}{m_L} \in [0,1].
\]
Thus the rival's repricing weakly raises the normalized current waiting rent whenever $m_L \ge m_H$, and raises it strictly when $m_L > m_H$. The original doubling case is the special case $\delta = 1$. The condition $m_L \ge m_H$ is
precisely the benchmark rival-enhanced-waiting condition (\ref{eq:rival-enhanced}) for this embedding.

Let initial vintage inventory be $B_0 > 0$. Before the rival reprices, the waiting-enabling inventory evolves as
\[
b_0(t) = B_0 - \frac{D}{2} t.
\]
If the rival reprices at date $s$ and the remaining seller stays at the legacy price until $t \ge s$, ordinary demand rises to $D$, so
\[
b_1(t;s) = B_0 - D t + \frac{D}{2} s.
\]
To keep this ordinary-demand path feasible throughout the timing horizon, impose the transparent sufficient restriction $B_0 - DT \ge 0$. Then $b_1(t;s) \ge B_0 - DT \ge 0$ for every $0 \le s \le t \le T$. A natural maximal robust-feasibility choice is $T = B_0/D$; more generally, any $T \le B_0/D$ is admissible.

The rival's repricing can also create an immediate customer-service or commitment exposure whose severity is greater when the remaining buffer is small. Let $Q > 0$ be a benchmark scale of redirected orders or commitments, let
$L_0 > 0$ and $c > 0$, and define
\[
\mathcal{L}(b) = L_0 + c \max\{Q - b, 0\}.
\]
The map is continuous and strictly positive, is weakly decreasing in the buffer, and is strictly decreasing for $b < Q$. It can summarize the present value of unfilled commitments, emergency procurement, or durable customer loss
from advertising a low price with little remaining stock. Importantly, $\mathcal{L}$ is an exposure cost, not an additional physical quantity subtracted from the ordinary-demand buffer path; ordinary redirected-demand depletion is already captured by $b_1$. Thus there is no double counting. In the general notation, $\Lambda(t) = \mathcal{L}\bigl(b_0(t)\bigr)$.

For the inventory clock to affect the penalty strictly throughout the horizon, one may impose $Q > B_0$; more generally, it is enough that $b_0(t) < Q$ on a positive-measure part of the equilibrium support. Without such a
restriction, $\mathcal{L}$ can be locally flat, in which case inventory remains a feasibility state but does not locally move the penalty.

Substituting these primitives into (\ref{eq:markov-sep-ode}), the equilibrium ODE becomes
\begin{equation}\label{eq:bertrand-ode}
\dot{F}(t) = a_B(t)\bigl[1 + \delta F(t)\bigr], \qquad
a_B(t) \equiv \frac{m_L D}{2\,\mathcal{L}\bigl(b_0(t)\bigr)}.
\end{equation}
For $\delta > 0$, solving (\ref{eq:bertrand-ode}) gives $1 + \delta F(t) = \exp\bigl\{\delta \int_{t_0}^{t} a_B(u)\, du\bigr\}$; for $\delta = 0$, $F(t) = \int_{t_0}^{t} a_B(u)\, du$. Define $\phi(\delta) = \log(1+\delta)/\delta$ for $\delta > 0$ and
$\phi(0) = 1$. The exact mass condition is $\int_0^{T} a_B(u)\, du > \phi(\delta)$, and the unique support start satisfies
\begin{equation}\label{eq:bertrand-support-start}
\int_{t_0}^{T} a_B(u)\, du = \phi(\delta).
\end{equation}

The key feedback is now explicit. Holding the legacy price earns a positive normalized rent. A rival's price increase redirects demand and, under $m_L > m_H$, raises that rent, creating the local second-mover flow advantage. The same redirected demand runs down vintage inventory faster and, when the remaining buffer is small, raises exposure to unfilled commitments, emergency procurement, backorders, and durable customer loss---the last-mover disadvantage. Inventory therefore affects timing both through feasibility and through the buffer-dependent exposure $\Lambda(t)$: the profitable activity itself runs down the state that disciplines delay. Because the cost shock is permanent, a return from the adjusted price to the legacy price lies outside this adjustment episode; later sales, cost reversals, or new repricing decisions can be treated as new timing problems.

The motivating retail language is familiar: ``10\% off while supplies last.'' The phrase explicitly ties a temporarily low price to a finite stock, which is precisely the buffer logic above. It does not require the stronger claim that a
retailer can never revise its plan; rather, it is consistent with a promotion or repricing episode that is costly to reopen once implemented. A frictionless benchmark in which every rival move triggers costless, instantaneous
replanning is therefore an analytically legitimate alternative, but need not be the institution represented by an inventory-limited retail offer.

Relative to adjacent work on inventory competition, the RWoA application isolates the timing role of the productive buffer. \citet{lippman1997competitive} and \citet{mahajan2001inventory} study strategic inventory competition and demand reallocation under stockouts. Those papers make inventory strategically important, but the strategic objective is inventory choice rather than the timing of exiting a profitable legacy state. Most importantly, in RWoA, inherited inventory is instead the endogenous clock: the same finite stock supports profitable low-price waiting and disciplines its duration.

\subsection{Competitive Maintenance}\label{sec:6.2}

Consider two firms operating revenue-producing machines, aircraft, engines, or production lines. Each has a remaining useful life buffer measured in hours, cycles, mileage, cumulative load, or deterioration state. With both assets in
operation, each earns $g_0 > 0$ and its buffer is depleted at a baseline rate $q_0 > 0$.

If one firm withdraws for maintenance, thus stopping production, the rival captures diverted demand. Its operating flow rises to $g_1 > g_0$, but higher utilization accelerates wear, so $q_1 > q_0$. If the remaining firm, now the
second and last mover, has too little useful life, it faces breakdown, emergency maintenance, regulatory violation, or costly downtime. These consequences generate a positive late-mover-penalty map $\mathcal{L}(b)$, typically with
$\mathcal{L}'(b) < 0$ on the relevant range.

The single-agent maintenance problem commonly produces a deterministic intervention rule. Strategic competition changes the timing logic. Neither firm wants to withdraw first because doing so concedes revenue to the demand-stealing rival. Yet remaining in operation after the rival withdraws becomes increasingly dangerous because the extra business consumes the remaining-life buffer faster, raising the risk of exhaustion and breakdown. Remaining useful life is therefore another literal productive buffer clock: operation earns the rent and simultaneously depletes the stock that limits how long it can continue. Once maintenance is initiated, return to operation after completion starts a new operating cycle rather than reversing the maintenance decision within the same episode.

\subsection{Other Productive Buffers}\label{sec:6.3}

The RWoA mechanism applies whenever the waiting state is supported by a finite, economically productive stock and rival action changes the waiting rent, depletion rate, or late-mover exposure. Strong candidates include nonrenewable
reserved AI compute or test capacity, remaining useful life or time-to-scheduled-maintenance for a machine or technology that is profitably running, and the stock of still-profitable legacy products that will ultimately be
replaced by newer versions. A productive buffer ticking as an endogenous clock must be economically necessary for profitable waiting, be depleted or degraded by that waiting, and create an exposure that cannot be undone costlessly
by moving an instant after the rival.

\section{Continuum Replication, Heterogeneous Local Markets, and Deterministic Aggregation}\label{sec:7}

The canonical RWoA is bilateral. For aggregate applications, the economically interesting setup seems to be the one with a continuum of heterogeneous disjoint local binary markets. Index local markets by $j \in [0,1]$. Market $j$ has
type $\theta(j) \in \Theta$, with cross-sectional type distribution $H$. Type summarizes local primitives---for example, inventory, demand, margins, depletion rates, or maintenance conditions---and is common to the two symmetric
players within that local market. Suppose the assumptions of Theorem \ref{thm1} hold for $H$-almost every type, generating a type-specific equilibrium CDF $F_\theta(t)$, and assume $(\theta,t) \mapsto F_\theta(t)$ is measurable. Conditional on
type, the two players in a local market draw independently from $F_\theta$. Across local markets, assume the pairs of idiosyncratic stopping draws are essentially pairwise independent in the sense required for the exact law of large
numbers.

Formally, the exact aggregation statement is understood on a rich Fubini extension in the sense of \citet{sun2006exact}. An ordinary product of a continuum index space and a probability space does not in general provide the jointly measurable essentially pairwise-independent process required for an exact law of large numbers; the Fubini-extension assumption supplies that measure-theoretic foundation.

Define the cross-sectional mean local switching CDF and the aggregate individual waiting share by
\[
F_H(t) = \int_{\Theta} F_\theta(t)\, H(d\theta), \qquad S(t) = 1 - F_H(t).
\]
A type-$\theta$ local market can be in three states: $WW$, with both players waiting; $WA/AW$, with exactly one waiting; and $AA$, with both having moved. Conditional independence within a local market gives the type-specific
probabilities $(1-F_\theta)^2$, $2F_\theta(1-F_\theta)$, and $F_\theta^2$.

\begin{proposition}[Heterogeneous continuum aggregation]\label{prop3}
In a continuum of heterogeneous local RWoA markets as above, individual stopping remains random but the aggregate state shares are deterministic and equal to
\[
\mu_{\mathrm{WW}}(t) = \int_{\Theta} \bigl[1 - F_\theta(t)\bigr]^2 H(d\theta),
\]
\[
\mu_{\mathrm{WA/AW}}(t) = \int_{\Theta} 2 F_\theta(t)\bigl[1 - F_\theta(t)\bigr] H(d\theta),
\]
\[
\mu_{\mathrm{AA}}(t) = \int_{\Theta} F_\theta(t)^2 H(d\theta).
\]
If $(\theta,t) \mapsto f_\theta(t)$ is jointly measurable, $F_\theta$ has density $f_\theta$ for $H$-almost every $\theta$, and $0 \le f_\theta(t) \le \psi(\theta)$ on the relevant interval for some $\psi \in L^1(H)$, then, for
almost every $t$, differentiation under the integral gives
\[
\dot{\mu}_{\mathrm{AA}}(t) = 2 \int_{\Theta} F_\theta(t) f_\theta(t)\, H(d\theta).
\]
\end{proposition}

\begin{proof}
On the Fubini extension described above, apply the exact law of large numbers to the local-market state indicators generated by the essentially pairwise independent market-level stopping pairs \citep{sun2006exact}. Conditional on type, the two stopping draws within a local market are independent, so the three local-state probabilities are the displayed binomial probabilities. Integrating those probabilities over the cross-sectional type distribution $H$ gives the
aggregate shares. The final formula follows by dominated differentiation.
\end{proof}

Heterogeneity is not innocuous. Let
\[
V_F(t) = \int_{\Theta} \bigl[F_\theta(t) - F_H(t)\bigr]^2 H(d\theta).
\]
Then the state shares admit the exact decomposition
\[
\mu_{\mathrm{AA}} = F_H^2 + V_F, \qquad \mu_{\mathrm{WW}} = (1 - F_H)^2 + V_F, \qquad \mu_{\mathrm{WA/AW}} = 2 F_H (1 - F_H) - 2 V_F.
\]
Thus, holding fixed the average fraction switched, cross-market heterogeneity puts more mass on local markets in the same state---both waiting or both moved---and less mass on split markets.

The continuum setup studied resembles mean-field environments but differs substantially from a Mean Field Game (MFG) in the modern Lasry--Lions or Huang--Malham\'e--Caines formulation. In those MFG models, individual optimization and
the population distribution are linked by an equilibrium consistency condition \citep{huang2006large,lasry2007mean}. Here the strategic fixed point is bilateral and is solved within each local market before aggregation, so the cross-sectional distribution does not feed back into local RWoA payoffs. This section is therefore an exact-LLN aggregation of solved local games, not an MFG equilibrium. It seems closer in spirit to the population-game/mean-dynamic language in which stochastic individual revisions generate deterministic aggregate shares \citep{sandholm2015population}.

The homogeneous replication is obtained as the special case $F_\theta(t) \equiv F(t)$, so $V_F(t) = 0$. Then
\begin{align*}
S(t) &= 1 - F(t), & \mu_{\mathrm{WW}}(t) &= [1-F(t)]^2, \\
\mu_{\mathrm{WA/AW}}(t) &= 2F(t)[1-F(t)], & \mu_{\mathrm{AA}}(t) &= F(t)^2.
\end{align*}
If $h(t) = f(t)/[1-F(t)]$ is the individual conditional switching hazard, the flow into the fully switched state is
$\dot{\mu}_{AA}(t) = 2F(t)f(t) = 2F(t)[1-F(t)]h(t)$. At the common endogenous support start $t_0$, $F(t_0) = 0$ and the indifference equation gives $f(t_0) = g_0(t_0)/\Lambda(t_0) > 0$. Interpreting derivatives at the support boundary
from the right,
\[
\dot{\mu}_{\mathrm{AA}}(t_0) = 0, \qquad \ddot{\mu}_{\mathrm{AA}}(t_0) = 2 f(t_0)^2 > 0.
\]
Thus, in the homogeneous benchmark, the share of fully switched markets has zero right-hand slope and a strictly positive right-hand second derivative at the onset of mixing, even though individual switching density is already
strictly positive there. This is a second-order aggregation effect: a local market reaches $AA$ only after two independent stopping events. The statement is an onset result, not a claim of global convexity. With heterogeneous
support starts, aggregate onset need not have this same second-order behavior; the general object is the cross-sectional integral of type-specific state probabilities.

More generally, let a bounded local outcome depend on both market type and local state, with values $y_{\mathrm{WW}}(t,\theta)$, $y_{\mathrm{WA}}(t,\theta)$, and $y_{\mathrm{AA}}(t,\theta)$. Its aggregate path is
\[
Y(t) = \int_{\Theta} \Bigl\{ [1-F_\theta(t)]^2 y_{\mathrm{WW}}(t,\theta) + 2F_\theta(t)[1-F_\theta(t)] y_{\mathrm{WA}}(t,\theta)
+ F_\theta(t)^2 y_{\mathrm{AA}}(t,\theta) \Bigr\} H(d\theta).
\]
In applications, $Y$ may be a transaction price, quantity, production status, capacity utilization, maintenance state, or another local observable. The bilateral RWoA determines each $F_\theta$; cross-sectional heterogeneity and the
exact law of large numbers then determine the aggregate path across all markets.

\section{Conclusions}\label{sec:8}

In this paper, we develop RWoA as a distinct timing game in which buffer-enabled waiting earns a positive current rent, a rival's move can raise that rent, and the player that remains in the waiting state after the rival moves is
simultaneously the second mover and the last mover, exposed to a discrete buffer-dependent late-mover penalty. RWoA therefore features neither the first-mover advantage of preemption nor the unconditional last-mover advantage of
classical WoA. In primitive terms, it reverses the waiting-flow sign of classical attrition and locally inverts the preemption force emphasized by \citet{fudenberg1985preemption}, while buffer-generated late-mover exposure prevents indefinite delay.

Our contribution has two parts. First, we combine strategic ingredients with familiar counterparts in the timing-games research---profitable delay, follower or second-mover advantages, mixed stopping, and disadvantages of being late---into the particular RWoA geometry. Second, and most importantly, we introduce the self-depleting productive buffer as the endogenous clock that endows RWoA with a concrete, and seemingly intuitive, economic mechanism. The buffer is not an exogenous time deadline: it makes the waiting state valuable and feasible, while its use progressively reduces the ability to remain in that state. A rival's move can simultaneously strengthen the current reward from waiting and make the clock run faster or increase the late-mover exposure. To our knowledge, prior timing-game work has not used a player-owned productive buffer in this dual enabling-and-disciplining role.

The buffer's dual role is the paper's central economic insight. The underlying stocks are familiar and observable: vintage inventory finances low-price waiting and measures how long it can continue; remaining useful life finances productive operation and measures how long maintenance can be held off. Similar clocks may arise from nonrenewable compute or test capacity and from exhaustible reserves supporting a legacy technology. Under transparent regularity and mass conditions, the bilateral game has a unique symmetric equilibrium in the stated ex-ante committed-date strategy class; that equilibrium is atomless and absolutely continuous, characterized by a Volterra equation and a shooting condition.

Aggregation is deliberately downstream from this bilateral core. A continuum of heterogeneous local binary markets generates deterministic aggregate state shares by integrating type-specific RWoA probabilities; heterogeneity itself
has an exact composition effect, increasing the mass of same-state local markets relative to the identical-market benchmark. This continuum extension admits a population-game mean-dynamic interpretation, but it is not a Mean Field Game equilibrium because, with the continuum of disjoint binary markets, the aggregate distribution does not feed back into the local strategic fixed point.

\begin{appendix}
\section{Details for the Volterra Construction}

\subsection{Existence and Uniqueness for a Fixed Support Start}\label{appA1}

Fix $\tau \in [0,T]$. By Assumption \ref{assump1}, on the triangle $\tau \le s \le t \le T$, $a$ and $K$ are continuous. For the estimates below, write $A = \lVert a \rVert_{\infty}$ and $B = \lVert K \rVert_{\infty}$. Here the sup norms are taken over the full domains $[0,T]$ for $a$ and $0 \le s \le t \le T$ for $K$, so $A$ and $B$ are independent of $\tau$. Define successive approximations
\[
f_\tau^{(0)}(t) = a(t),
\]
\[
f_\tau^{(m+1)}(t) = a(t) + \int_{\tau}^{t} K(t,s) f_\tau^{(m)}(s)\, ds.
\]
\[
\bigl| f_\tau^{(m+1)}(t) - f_\tau^{(m)}(t) \bigr| \le A\,\frac{B^{m+1} (t-\tau)^{m+1}}{(m+1)!}.
\]
The displayed factorial bound is uniformly summable on the interval, so the successive-approximation sequence is uniformly Cauchy and converges uniformly to a solution of the Volterra equation. If $f$ and $\tilde{f}$ are two solutions, their difference satisfies $|f(t) - \tilde{f}(t)| \le B \int_{\tau}^{t} |f(s) - \tilde{f}(s)|\, ds$; Gronwall's inequality gives uniqueness. Finally, $a > 0$ and $K \ge 0$ imply that the solution is strictly positive.

The fixed-point and factorial-bound argument is the standard successive-approximation construction for a linear Volterra equation of the second kind; see \citet[Chapter 1]{brunner2017volterra} and, in the more general framework used here, \citet[Theorem 2.1.1, pp.~86--87]{kolokoltsov2019differential}.

\subsection{Strict Monotonicity and Continuity of the Shooting Map}\label{appA2}

Let $\tau_1 < \tau_2$ and write $f_1 = f_{\tau_1}$ and $f_2 = f_{\tau_2}$. For $t \ge \tau_2$,
\begin{equation}\label{eq:shooting-diff}
f_1(t) - f_2(t) = \int_{\tau_1}^{\tau_2} K(t,s) f_1(s)\, ds + \int_{\tau_2}^{t} K(t,s)\bigl[f_1(s) - f_2(s)\bigr] ds.
\end{equation}
\[
M(\tau_1) - M(\tau_2) = \int_{\tau_1}^{\tau_2} f_1(s)\, ds + \int_{\tau_2}^{T} \bigl[f_1(s) - f_2(s)\bigr] ds > 0.
\]
Because the forcing term in formula (\ref{eq:shooting-diff}) is nonnegative and $K$ is nonnegative, Picard iteration (equivalently, positivity of the Volterra resolvent) implies $f_1(t) \ge f_2(t)$ on $[\tau_2,T]$. The displayed identity then gives strict
monotonicity of $M$.

To establish continuity, the global bounds fixed in Appendix \ref{appA1} and Gronwall's inequality give, uniformly in $\tau$,
\[
0 < f_\tau(t) \le C := A e^{BT}, \qquad t \in [\tau,T].
\]
For $\tau_1 < \tau_2$, equation (\ref{eq:shooting-diff}) therefore implies
\[
|f_1(t) - f_2(t)| \le BC(\tau_2 - \tau_1) + B \int_{\tau_2}^{t} |f_1(s) - f_2(s)|\, ds.
\]
Applying Gronwall's inequality yields
\[
\lVert f_1 - f_2 \rVert_{C([\tau_2,T])} \le B C e^{BT} |\tau_2 - \tau_1|.
\]
Consequently,
\[
0 < M(\tau_1) - M(\tau_2) \le \bigl[C + TBCe^{BT}\bigr] |\tau_2 - \tau_1|.
\]
Hence $M$ is Lipschitz continuous, as well as strictly decreasing.

\subsection{Verification}\label{appA3}

For the selected $t_0$, equation (\ref{eq:payoff-diff}) together with the Volterra equation gives a zero payoff derivative for almost every date in $[t_0,T]$. Since expected payoff is absolutely continuous, it is constant on this support. Before $t_0$, $F = 0$ and the derivative is $e^{-rt} g_0(t) > 0$. Hence payoff is strictly increasing up to $t_0$. Continuity handles the boundary points. This verifies every pure deviation; linearity of expected payoff verifies every mixed deviation. Together with Proposition \ref{prop2}, Lemma \ref{lem1}, and the shooting argument, this also verifies uniqueness among symmetric equilibria in the committed-date strategy class.

\subsection{Two Special-Case Calculations}

For the pricing application, formula (\ref{eq:bertrand-ode}) is $\dot{F} = a_B(t)(1 + \delta F)$. If $\delta > 0$, then $d\log(1+\delta F)/dt = \delta a_B(t)$; using $F(t_0) = 0$ and $F(T) = 1$ gives
$\int_{t_0}^{T} a_B(u)\, du = \log(1+\delta)/\delta$. If $\delta = 0$, direct integration gives $\int_{t_0}^{T} a_B(u)\, du = 1$. This proves formula (\ref{eq:bertrand-support-start}) with the continuous convention $\phi(0) = 1$.

For the homogeneous aggregation benchmark, $\mu_{AA}(t) = F(t)^2$ and $\dot{\mu}_{AA}(t) = 2F(t)f(t)$ for $t > t_0$. Since $F(t_0) = 0$, continuity of the Volterra density gives $F(t_0+h)/h \to f(t_0)$ and $f(t_0+h) \to f(t_0)$ as
$h \downarrow 0$. Thus the right-hand first derivative of $\mu_{AA}$ at $t_0$ is zero, while its right-hand second derivative is $2f(t_0)^2 > 0$. No differentiability of $f$ is needed for this onset calculation.
\end{appendix}


\bibliographystyle{econsoc} 
\bibliography{bibliography}  

\end{document}